\documentclass[letterpaper,10pt,conference]{ieeeconf}
\IEEEoverridecommandlockouts
\usepackage{amsmath,amssymb,amsfonts,mathtools}
\usepackage{amsthm}
\usepackage{bm}
\usepackage{cite}
\usepackage{algorithm}
\usepackage{algorithmic}
\usepackage{graphicx}
\usepackage{booktabs}

\newtheorem{assumption}{Assumption}
\newtheorem{theorem}{Theorem}
\newtheorem{lemma}{Lemma}

\newtheorem{corollary}{Corollary}

\newcommand{\B}{\mathcal{B}}
\newcommand{\R}{\mathbb{R}}
\newcommand{\U}{\mathcal{U}}
\newcommand{\D}{\mathcal{D}}
\newcommand{\Ssafe}{\mathcal{S}}
\newcommand{\Ksafe}{\mathcal{K}_B}
\newcommand{\Xf}{\mathcal{X}_f}

\newcommand{\Bball}{\mathbb{B}}
\newcommand{\tr}{\mathrm{tr}}

\title{\LARGE\bfseries Safety-Aware Dimension Selection for Reduced-Order Model Predictive Control}

\author{Amirhossein Asadi$^{1}$ and Nitin Sharma$^{2}$%
\thanks{$^{1}$A. Asadi is with the Department of Mechanical and Aerospace Engineering, NC State University, Raleigh, NC 27606 USA (e-mail: aasadi@ncsu.edu).}%
\thanks{$^{2}$N. Sharma is with the UNC/NC State Joint Department of Biomedical Engineering, NC State University, Raleigh, NC 27606 USA (e-mail: nsharm23@ncsu.edu). Corresponding author: Nitin Sharma.}%
}

\begin{document}
\maketitle
\thispagestyle{empty}
\pagestyle{empty}

\begin{abstract}
Model reduction can substantially reduce the computational burden of model predictive control (MPC) for high-dimensional systems, but a reduced dimension chosen for performance may be insufficient for safety-critical operation. This paper introduces a safety-aware dimension-selection framework in which the prediction model adapts to the order of the model required to maintain a certified finite-horizon safety condition. Dynamic active subspaces provide a performance-oriented ordering of reduced coordinates, while finite-horizon error bounds robustly enforce discrete-time control barrier function constraints. The controller selects the smallest prediction order for which the resulting safety-critical MPC problem remains feasible. We establish full-order horizon safety, recursive feasibility, and closed-loop invariance of the safe set under the proposed reduced-order formulation. The proposed approach therefore adapts the prediction complexity to safety requirements, avoiding unnecessary high-dimensional prediction during benign operation.

\end{abstract}


\section{Introduction}

Model predictive control (MPC) is attractive for constrained nonlinear systems because it optimizes performance while explicitly accounting for future dynamics and constraints. Its main limitation in high-dimensional applications is computational, including applications such as whole-body robotic control, power and energy networks, and biomedical systems with coupled human--robot or neuromuscular dynamics \cite{schurig2024,lorenzetti2022}. The plant model is propagated repeatedly over the prediction horizon, and direct transcription methods introduce state variables and dynamic constraints whose number grows with model dimension. Model-order reduction addresses this difficulty by replacing the full-order prediction with a lower-dimensional representation. Active subspaces provide sensitivity-based coordinates for low-dimensional ridge approximations of scalar quantities of interest \cite{constantine2015,constantineRidge2017}, and time-dependent variants allow the important directions to evolve with operating conditions \cite{constantine2017}. In control, low-dimensional subspace methods have been used to reduce MPC decision spaces while retaining feasibility and stability guarantees \cite{pan2023,schurig2024}, and dynamic active subspaces (DyAS) have been used for model-predictive allocation in over-actuated systems \cite{singh2024}. Active-subspace state abstractions have also been paired with formal error guarantees \cite{awan2025}.

Safety imposes different requirements when reduced models are used for safety-critical control \cite{cohen2024romcbf}. A DyAS dimension selected based on nominal performance sensitivity may be sufficient for the task objective but too small for accurate prediction near a safety-critical regime. In contrast, using a fixed dimension sized for worst-case conditions can sacrifice much of the computational savings during benign operation \cite{dones2010}. Discrete-time control barrier functions (DCBFs) provide a natural language for this tradeoff because they express safety as a one-step inequality and can be embedded directly in MPC \cite{zeng2021,katriniok2024}. Reduced-order MPC, meanwhile, uses computable approximation bounds to tighten constraints and certify a low-dimensional controller \cite{lohning2014,kartmann2026}, while adaptive model-reduction methods update or enrich the reduced representation when a static approximation becomes inadequate \cite{chellappa2020,peherstorfer2015}.

The central question of this paper is: 1) when does a reduced dimension, selected based on performance sensitivity, sufficiently predict safety? and 2) when are additional coordinates required? Rather than imposing a fixed approximation tolerance, we assess each candidate DyAS order through the resulting finite-horizon safety certificate. Starting from the conventional performance order, we increase the prediction dimension only when the associated reduction error prevents feasibility of the robust DCBF-MPC problem. Two modeling choices are essential to make this idea mathematically precise. First, an active subspace supplies sensitivity coordinates; it does not by itself imply that the physical state is well approximated by the orthogonal projection \cite{constantineRidge2017}. We therefore use a separate decoder to reconstruct physical states from the DyAS coordinates, and all reconstruction quality enters the certificate explicitly. Second, because the full-order state is measured at every MPC update, the first DCBF constraint is evaluated exactly using the full-order one-step map. The reduced model is used only for future prediction. Hence, safety-driven order adaptation is required for maintaining a feasible future certificate and terminal reachability, not for evaluating instantaneous safety at the measured state.

The proposed architecture is intentionally simple. For each candidate DyAS order, a low-dimensional predictor and decoder generate the future nominal trajectory. The exact first-step mismatch initializes a scalar prediction-error envelope, which is propagated over the remaining horizon using a certified one-step model defect and an incremental bound on the full-order dynamics. This envelope tightens all future DCBF constraints and enforces that each uncertainty set remains inside a fixed certification domain. A single basis-independent full-order terminal set, computed offline, closes the recursive-feasibility argument, and its terminal feedback is chosen from the full-order DCBF-admissible input set. Under the deterministic plant model, a certified input sequence may also be executed for multiple steps if a later optimization fails.

The contributions of this paper are threefold. First, we develop a finite-horizon error certification framework that quantifies the full-order prediction uncertainty associated with each reduced DyAS order and uses it to robustly enforce future DCBF constraints. Second, we introduce a safety-aware dimension-selection rule that chooses the smallest prediction order, no lower than the performance-oriented DyAS order, for which the safety-critical MPC problem remains feasible. Third, we establish full-order horizon safety, recursive feasibility, and closed-loop invariance of the safe set, including the case where a previously certified input sequence is executed over multiple steps. Together, these results allow the prediction dimension to adapt to the safety requirements of the current operating condition rather than to a fixed approximation tolerance.

\section{Problem Formulation}

\subsection{System Dynamics and Safety}

Consider the sampled nonlinear system
\begin{equation}
    x_{k+1}=F(x_k,u_k),
    \qquad x_k\in\R^n,\quad u_k\in\U\subset\R^m,
    \label{eq:fom}
\end{equation}
where $F:\R^n\times\U\rightarrow\R^n$ is the discrete-time state-transition map and $\U$ is compact. The plant model is deterministic; the present analysis excludes exogenous disturbances and plant--model mismatch. We use this assumption explicitly when we retain a previously certified input sequence as a backup. We also fix a compact certification domain $\D\subset\R^n$ containing the intended operating region. Let
\begin{equation}
    \Ssafe:=\{x\in\D:h(x)\ge0\},
    \label{eq:safeset}
\end{equation}
where $h:\R^n\rightarrow\R$ is continuous. For $x\in\D$, define the discrete-time barrier-admissible input set
\begin{equation}
    \Ksafe(x):=
    \left\{
    u\in\U:
    F(x,u)\in\D,\;
    h(F(x,u))\ge\beta h(x)
    \right\},
    \label{eq:Ksafe}
\end{equation}
where $\beta\in[0,1)$. Following standard discrete-time CBF terminology, $h$ is a DCBF on $\Ssafe$ when $\Ksafe(x)$ is nonempty for every $x\in\Ssafe$; any input in $\Ksafe(x)$ maps a safe state to another state in $\Ssafe$. The analysis does not require $\Ksafe(x)$ to be nonempty for every \(x\in S\); instead, the controller enforces admissibility explicitly.

\begin{assumption}
\label{ass:regularity}
The map $F$ is jointly continuous on $\D\times\U$ and there exist constants $\rho\ge0$ and $L_h\ge0$ such that, for all $x,y\in\D$ and $u\in\U$,
\begin{align}
    \|F(x,u)-F(y,u)\|&\le\rho\|x-y\|,
    \label{eq:incremental}\\
    |h(x)-h(y)|&\le L_h\|x-y\|.
    \label{eq:hLip}
\end{align}
\end{assumption}

Because $F$ and $h$ are continuous and $\U$ is compact, $\Ksafe(x)$ is a closed subset of $\U$ and is therefore compact whenever nonempty, and a backup input can also be selected from it using any continuous optimization criterion. 

\subsection{Dynamic Active Subspaces}

Let $J_p:\mathbb{R}^n\times\mathbb{N}_0\rightarrow\mathbb{R}$ be a differentiable scalar performance observable. A streaming sensitivity covariance $C_{p,k}\in\mathbb{R}^{n\times n}$ is updated as
\begin{equation}
    C_{p,k+1}
    =(1-\eta)C_{p,k}
    +\eta\,\nabla J_p(x_k,k)\nabla J_p(x_k,k)^\top,
    \label{eq:covariance}
\end{equation}
where $\eta\in(0,1]$. We eigendecompose $C_{p,k}$ as
\begin{equation}
    C_{p,k}=W_k\Lambda_kW_k^\top,
    \qquad
    \lambda_{1,k}\ge\cdots\ge\lambda_{n,k}\ge0,
    \label{eq:eig}
\end{equation}
where $\Lambda_k=\mathrm{diag}(\lambda_{1,k},\ldots,\lambda_{n,k})$ contains the eigenvalues of $C_{p,k}$, and $W_k=[w_{1,k},\ldots,w_{n,k}]$ is the orthogonal matrix of corresponding eigenvectors. Define
\begin{equation}
    W_{q,k}:=[w_{1,k},\ldots,w_{q,k}],
    \qquad q=1,\ldots,n.
    \label{eq:basis}
\end{equation}
For a prescribed performance-energy fraction $\eta_p\in(0,1)$,
\begin{equation}
    q_p(k):=
    \min\left\{
    q:
    \frac{\sum_{i=1}^{q}\lambda_{i,k}}
         {\sum_{i=1}^{n}\lambda_{i,k}}
    \ge\eta_p
    \right\},
    \label{eq:qperf}
\end{equation}
whenever $\tr(C_{p,k})>0$. Thus, $q_p(k)$ is the reduced prediction order selected purely from the performance-sensitivity criterion, without considering safety constraints. Otherwise, a prescribed minimum order is used. The DyAS basis defines sensitivity coordinates but is \emph{not} assumed to be a state-approximation basis. For a candidate $q\ge q_p(k)$, hold $W_q:=W_{q,k}$ fixed over the current prediction horizon and define
\begin{equation}
    z_{0|k}=W_q^\top x_k,
    \qquad
    z_{j+1|k}=F_q(z_{j|k},u_{j|k}).
    \label{eq:rom}
\end{equation}
The reduced predictor $F_q:\R^q\times\U\rightarrow\R^q$ may be constructed by projection, regression, or another model-reduction procedure; the specific construction used in the numerical study is given in Section~V. A separate decoder
\begin{equation}
    \hat x_{j|k}=R_q(z_{j|k}),
    \qquad R_q:\R^q\rightarrow\R^n,
    \label{eq:decoder}
\end{equation}
reconstructs a physical nominal state from the active coordinates.

\begin{assumption}
\label{ass:rom}
For every candidate order $q$, the maps $F_q$ and $R_q$ are continuous on the relevant compact sets, and there exists a computable nonnegative defect bound $\delta_q$ such that
\begin{equation}
    \|F(R_q(z),u)-R_q(F_q(z,u))\|
    \le\delta_q(z,u)
    \label{eq:defect}
\end{equation}
whenever $R_q(z)\in\D$ and $u\in\U$. For $q<n$, evaluation of $F_q$, $R_q$, and $\delta_q$
is assumed cheaper than repeated full-order propagation.
At full order $q=n$, with $W_n=W_{n,k}$, set
$R_n(z)=W_nz$, $F_n(z,u)=W_n^\top F(W_nz,u)$,
and $\delta_n(z,u)\equiv0$.
\end{assumption}

The decoder explicitly separates performance sensitivity from state reconstruction. If the active coordinates do not reconstruct future physical states accurately at low order, the certified error envelope developed next becomes large and that order is rejected by the future safety constraints.
\section{Safety Certification and Order Selection}

\subsection{Finite-Horizon Error Bound}

For a fixed prediction horizon $N\in\mathbb{N}$ and candidate input sequence, define the conceptual full-order prediction
\begin{equation}
    x_{0|k}=x_k,
    \qquad
    x_{j+1|k}=F(x_{j|k},u_{j|k}).
    \label{eq:fompred}
\end{equation}
Only the first full-order step is evaluated online inside a reduced candidate. For a given reduced prediction, define the \emph{minimal certified envelope} recursively by
\begin{align}
    \bar r_{1,q}
    &:={}
    \|F(x_k,u_{0|k})-R_q(z_{1|k})\|,
    \label{eq:r1}\\
    \bar r_{j+1,q}
    &:={}
    \rho\bar r_{j,q}+\delta_q(z_{j|k},u_{j|k}),
    \quad j=1,\ldots,N-1.
    \label{eq:envelope}
\end{align}
Thus $\bar r_{j,q}$ is an auxiliary quantity determined by $(\mathbf z,\mathbf u)$, not a free optimization variable. Let $\Bball(r):=\{e:\|e\|\le r\}$, and $\oplus$ denotes the Minkowski sum.

\begin{lemma}
\label{lem:error}
Under Assumptions~\ref{ass:regularity}--\ref{ass:rom}, suppose
\begin{equation}
    R_q(z_{j|k})\oplus\Bball(\bar r_{j,q})\subseteq\D,
    \qquad j=1,\ldots,N.
    \label{eq:domaincontain}
\end{equation}
Then
\begin{equation}
    \|x_{j|k}-R_q(z_{j|k})\|\le\bar r_{j,q},
    \qquad j=1,\ldots,N.
    \label{eq:errorbound}
\end{equation}
\end{lemma}

\begin{proof}
The claim holds at $j=1$ by \eqref{eq:r1}. If it holds at step $j$, then \eqref{eq:domaincontain} places both $x_{j|k}$ and $R_q(z_{j|k})$ in $\D$. Hence
\begin{align}
&\|x_{j+1|k}-R_q(z_{j+1|k})\|\nonumber\\
&\quad\le
\|F(x_{j|k},u_{j|k})-F(R_q(z_{j|k}),u_{j|k})\|\nonumber\\
&\qquad+
\|F(R_q(z_{j|k}),u_{j|k})-R_q(F_q(z_{j|k},u_{j|k}))\|\nonumber\\
&\quad\le \rho\bar r_{j,q}+\delta_q(z_{j|k},u_{j|k})
=\bar r_{j+1,q},
\end{align}
which completes the induction.
\end{proof}

The envelope is finite-horizon only. Therefore $\rho$ may exceed one; this can make low-order certificates conservative with prediction depth, but it creates no terminal invariance contradiction. 

\subsection{Robust DCBF Constraints}

At the measured state, impose the exact full-order condition
\begin{equation}
    u_{0|k}\in\Ksafe(x_k),
    \label{eq:exactfirst}
\end{equation}
which checks both the DCBF inequality and successor membership in $\D$ without any reduced-state approximation.

For $j\ge1$, Lemma~\ref{lem:error} and \eqref{eq:hLip} give
\begin{equation}
    h(R_q(z_{j|k}))-L_h\bar r_{j,q}
    \le h(x_{j|k})
    \le h(R_q(z_{j|k}))+L_h\bar r_{j,q}.
    \label{eq:hinterval}
\end{equation}
Accordingly, for $j=1,\ldots,N-1$, define
\begin{align}
    \Gamma_{j,q}
    &:={}
    h(R_q(z_{j+1|k}))-L_h\bar r_{j+1,q}
    \nonumber\\
    &\quad-
    \beta\big(h(R_q(z_{j|k}))+L_h\bar r_{j,q}\big).
    \label{eq:budget}
\end{align}
The robust future DCBF constraints are
\begin{equation}
    \Gamma_{j,q}\ge0,
    \qquad j=1,\ldots,N-1.
    \label{eq:futurecbf}
\end{equation}
Thus order is not used to certify the measured state; it determines whether the reduced predictor is accurate enough to support a recursively feasible future certificate.

\subsection{Minimum Certified Order}

Let $\mathcal F_q(x_k)$ denote the feasible set of all control and reduced-state trajectories that satisfy the order-$q$ MPC constraints introduced in Section~IV. Define
\begin{equation}
    q^\star(k)
    :=\min\left\{
    q\in\{q_p(k),\ldots,n\}:\mathcal F_q(x_k)\neq\emptyset
    \right\},
    \label{eq:qstar}
\end{equation}
which can be interpreted as minimum prediction dimension needed for future certified feasibility.

\begin{theorem}
\label{thm:horizon}
Under Assumptions~\ref{ass:regularity}--\ref{ass:rom}, every feasible point of $\mathcal F_q(x_k)$ generates, under the same input sequence, a full-order prediction satisfying
\begin{align}
    x_{j|k}&\in\D,
    &&j=0,\ldots,N,
    \label{eq:horizondomain}\\
    h(x_{j+1|k})&\ge\beta h(x_{j|k}),
    &&j=0,\ldots,N-1.
    \label{eq:horizonchain}
\end{align}
If additionally $x_k\in\Ssafe$, then $h(x_{j|k})\ge\beta^j h(x_k)\ge0$ for all $j$. If $\mathcal F_n(x_k)\neq\emptyset$, the minimum in \eqref{eq:qstar} exists.
\end{theorem}

\begin{proof}
See Appendix.
\end{proof}
The first DCBF step is exact; the future steps are certified from the reduced prediction and its error envelope.

\section{Safety-Critical Reduced-Order MPC}

\subsection{Safe Terminal Set}

The safe terminal set is defined once in the full-order state space and is independent of the DyAS basis.

\begin{assumption}
\label{ass:terminal}
There exist a nonempty compact set $\Xf\subset\Ssafe$ and a terminal feedback $\kappa_f:\Xf\rightarrow\U$ such that
\begin{equation}
    \kappa_f(x)\in\Ksafe(x),
    \qquad
    F(x,\kappa_f(x))\in\Xf,
    \quad \forall x\in\Xf.
    \label{eq:terminal}
\end{equation}
\end{assumption}

Thus Assumption~\ref{ass:terminal} is a local DCBF-invariance condition on the backup region $\Xf$. The terminal controller is DCBF admissible and keeps $\Xf$ invariant. The set and controller are synthesized once offline and are not recomputed as the DyAS basis changes.

\subsection{Order-Adaptive MPC Formulation}

For each candidate $q\ge q_p(k)$, define the feasible set $\mathcal F_q(x_k)$ by the constraints below and seek the performance infimum
\begin{subequations}
\label{eq:mpc}
\begin{align}
J_q^{\inf}(x_k):=
\inf_{\mathbf z,\mathbf u}\;&
\sum_{j=0}^{N-1}\ell_q(z_{j|k},u_{j|k})
+V_{f,q}(z_{N|k})
\label{eq:cost}\\
\mathrm{s.t.}\quad
&z_{0|k}=W_q^\top x_k,
\label{eq:mpcz0}\\
&z_{j+1|k}=F_q(z_{j|k},u_{j|k}),
\; j=0,\ldots,N-1,
\label{eq:mpcdyn}\\
&u_{j|k}\in\U,
\; j=0,\ldots,N-1,
\label{eq:mpcu}\\
&R_q(z_{j|k})\oplus\Bball(\bar r_{j,q})\subseteq\D,
\; j=1,\ldots,N,
\label{eq:mpcdomain}\\
&u_{0|k}\in\Ksafe(x_k),
\label{eq:mpcfirst}\\
&\Gamma_{j,q}\ge0,
\quad j=1,\ldots,N-1,
\label{eq:mpcfuture}\\
&R_q(z_{N|k})\oplus\Bball(\bar r_{N,q})\subseteq\Xf.
\label{eq:mpcterm}
\end{align}
\end{subequations}
where $\ell_q:\R^q\times\U\to\R_{\ge0}$ and $V_{f,q}:\R^q\to\R_{\ge0}$ are, respectively,
a continuous stage cost and terminal cost that encode the tracking objective, and the radii $\bar r_{j,q}$ are given by the deterministic envelope \eqref{eq:r1}--\eqref{eq:envelope}. The guarantees apply to any feasible point; attainment of the infimum is not required. The containments in \eqref{eq:mpcdomain} and \eqref{eq:mpcterm} are exact set-inclusion conditions and may be replaced in implementation by any sufficient tightening. At $q=n$, Assumption~\ref{ass:rom} and \eqref{eq:r1}--\eqref{eq:envelope} give $\bar r_{j,n}=0$, so \eqref{eq:mpc} reduces to the exact full-order DCBF-MPC problem with terminal set $\Xf$.

The minimum order $q^\star(k)$ in \eqref{eq:qstar} is a set-theoretic quantity. In implementation, candidate orders are tested sequentially from $q_p(k)$ to $n$. Since a local nonlinear-programming solver cannot certify global infeasibility, let $\hat q(k)$ denote the smallest tested order for which a returned candidate is successfully verified; hence $\hat q(k)$ may exceed $q^\star(k)$. Failure to certify a feasible low order is conservative with respect to safety. A returned candidate is accepted only after independently reconstructing the reduced trajectory, recomputing the error envelope \eqref{eq:r1}--\eqref{eq:envelope}, and verifying the input, domain, DCBF, and terminal constraints in \eqref{eq:mpcu}--\eqref{eq:mpcterm}. Numerical verification accounts for reconstruction, envelope, and constraint-evaluation errors beyond the solver tolerance. If no candidate is verified, the controller follows the stored backup sequence and then applies $\kappa_f$ in $\Xf$.

% Positive inset keeps the algorithm rule inside the top margin.
\begin{figure}[t]
\vspace*{5pt}
\begin{minipage}{\linewidth}
\begin{algorithm}[H]
\caption{Safety-Aware Order-Adaptive DyAS MPC}
\label{alg:main}
\begin{algorithmic}[1]
\STATE Measure $x_k$, compute the DyAS basis \eqref{eq:basis}, and determine $q_p(k)$ from \eqref{eq:qperf}.
\FOR{$q=q_p(k),\ldots,n$}
    \STATE Construct $F_q$, $R_q$, and $\delta_q$ as in \eqref{eq:rom}--\eqref{eq:defect}.
    \STATE Solve \eqref{eq:mpc} at order $q$ to obtain a candidate feasible trajectory.
    \IF{a candidate is returned}
        \STATE Recompute \eqref{eq:r1}--\eqref{eq:envelope} and verify \eqref{eq:mpcu}--\eqref{eq:mpcterm}.
        \IF{verified}
            \STATE Set $\hat q(k)\leftarrow q$, store $\{u_{j|k}\}_{j=1}^{N-1}$, and apply $u_{0|k}$.
            \STATE \textbf{break}
        \ENDIF
    \ENDIF
\ENDFOR
\IF{no candidate is verified}
    \STATE Apply the next stored backup input, or $\kappa_f(x_k)$ from \eqref{eq:terminal} if the queue is empty.
\ENDIF
\end{algorithmic}
\end{algorithm}
\end{minipage}
\end{figure}

\subsection{Recursive Feasibility and Closed-Loop Safety}

\begin{theorem}
\label{thm:recursive}
Suppose Assumptions~\ref{ass:regularity}--\ref{ass:terminal} hold and a feasible point of $\mathcal F_q(x_k)$ is certified at time $k$. If the first $\ell\in\{1,\ldots,N\}$ inputs of that certified sequence are applied to the deterministic plant, then $\mathcal F_n(x_{k+\ell})\neq\emptyset$. In particular, if $\ell=N$, then $x_{k+N}\in\Xf$.
\end{theorem}

\begin{proof}
See Appendix.
\end{proof}
The unused portion of a certified horizon can be used as a backup sequence for up to $N$ consecutive steps, after which the state reaches $\Xf$. The case $\ell=1$ gives the standard one-step recursive-feasibility result, while $\ell=N$ shows that the stored sequence reaches $\Xf$, where $\kappa_f$ can continue the safety certificate.

\begin{corollary}
\label{cor:safety}
Suppose $x_0\in\Ssafe$ and an initial feasible candidate exists. If the controller applies either a newly certified input, the next stored backup input, or $\kappa_f(x_k)$ in $\Xf$, then $x_k\in\Ssafe$ for all $k\ge0$.
\end{corollary}

\begin{proof}
A newly certified input satisfies \eqref{eq:mpcfirst}. During backup execution, Theorem~\ref{thm:horizon} guarantees the stored trajectory remains in $\Ssafe$, and after horizon exhaustion Theorem~\ref{thm:recursive} places the state in $\Xf$, where \eqref{eq:terminal} makes $\kappa_f$ safe. Hence each control branch preserves $\Ssafe$, and induction gives the result.
\end{proof}

\section{Numerical Evaluation}

\begin{figure}[t]
    \vspace*{5pt}
    \centering
    \begin{minipage}[t]{\linewidth}
        \centering
        \includegraphics[width=\linewidth]{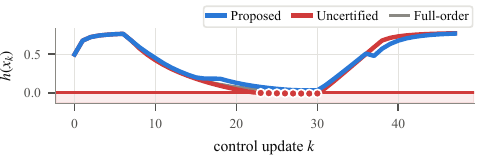}\par
        \small (a)
    \end{minipage}\par\medskip
    \begin{minipage}[t]{\linewidth}
        \centering
        \includegraphics[width=\linewidth]{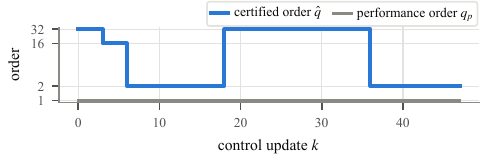}\par
        \small (b)
    \end{minipage}\par\medskip
    \begin{minipage}[t]{\linewidth}
        \centering
        \includegraphics[width=\linewidth]{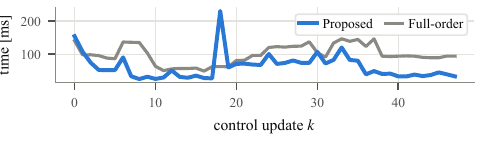}\par
        \small (c)
    \end{minipage}
    \caption{Closed-loop results: (a) physical barrier value for the proposed adaptive controller, uncertified DyAS, and full-order MPC; (b) certified and performance-oriented prediction orders; and (c) computation time for the proposed adaptive controller and full-order MPC.}
    \label{fig:closed-loop-results}
\end{figure}

We compare the proposed controller with uncertified DyAS, fixed-order certified MPC
($q=32$), and full-order MPC on
$x^+=Ax+bu+c\tanh(g^\top x)+\B(x,x)$, with
$n=160$, $|u|\le1$, and horizon $N=6$. Here
$A=\mathrm{blkdiag}(0.96,A_t)$, with $A_t$ random and rescaled so that
$\|A_t\|_2\le0.88$; $b=0.25e_1$;
$c=(0.06,c_t^\top)^\top$ with $\|c_t\|_2=0.12$; $g$ is a fixed unit vector; and the
random symmetric tensor is rescaled so that
$\|\B(a,b)\|_2\le0.04\|a\|_2\|b\|_2$.
On $\D=\{x:\|x\|_2\le1.5\}$, these bounds justify $\rho=1.215$. Safety is defined by
$h(x)=1-\|x\|_2$, with $L_h=1$. Using training-state covariance $\Sigma$, we set
$D_q=\Sigma W_q(W_q^\top\Sigma W_q)^{-1}$, $R_q(z)=D_qz$, and
$F_q(z,u)=W_q^\top F(D_qz,u)$, with projected coefficients assembled algebraically.
The residual is factored exactly as
$F(D_qz,u)-D_qF_q(z,u)=Q_qT_q\phi_q(z,u)$, where
$Q_q^\top Q_q=I$ and
$\phi_q=[z^\top,u,\tanh(g^\top D_qz),(z\otimes z)^\top]^\top$,
giving $\delta_q(z,u)=\|T_q\phi_q(z,u)\|_2$. The terminal set
$\Xf=\{x:\|x\|_2\le0.4\}$ uses
$\kappa_f(x)=(-0.46x_1-0.06\tanh(g^\top x))/0.25$; its
contraction factor is below $0.908$ and input magnitude below $0.8272$, establishing
terminal invariance and DCBF admissibility. The reference ramps from $0.20$ to $0.95$,
holds, and then relaxes; $q_p=1$ throughout. Figures~\ref{fig:closed-loop-results}
and~\ref{fig:order-feasibility} show a representative run, while
Fig.~\ref{fig:solve-time-comparison} aggregates five random initial conditions with
$48$ updates each ($240$ updates per controller).

Figure~\ref{fig:closed-loop-results}(a) compares the barrier value along a representative closed-loop trajectory. The proposed controller closely follows full-order MPC and maintains $h(x_k)\geq0$, while uncertified DyAS repeatedly crosses the safety boundary. This shows that performance-based order selection alone can be insufficient for safety-critical prediction. Figure~\ref{fig:closed-loop-results}(b) shows the prediction orders. The performance order remains at $q_p=1$, while the certified order varies with operating condition, reaching $32$ near the safety-critical interval and returning to lower values afterward. The increase occurs before the uncertified trajectory reaches the boundary, showing that selection is driven by future certification rather than the instantaneous barrier value. Figure~\ref{fig:closed-loop-results}(c) shows computation time. The proposed controller is generally faster than full-order MPC at low orders, while its cost increases with higher orders and additional order-search effort.

\begin{figure}[t]
    \vspace*{5pt}
    \centering
    \includegraphics[width=\linewidth]{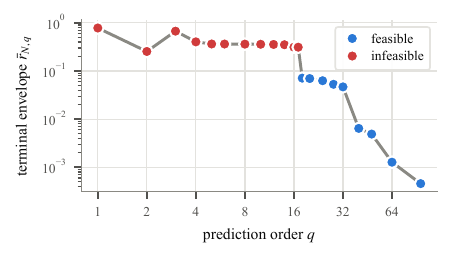}
    \caption{Terminal prediction-error envelope versus prediction order at a representative safety-critical update, with feasible and infeasible candidates indicated.}
    \label{fig:order-feasibility}
\end{figure}

Figure~\ref{fig:order-feasibility} shows the terminal prediction-error envelope $\bar r_{N,q}$ across candidate orders at a safety-critical update. Low-order candidates produce larger uncertainty and are infeasible, while higher orders eventually reduce the envelope enough to satisfy the certified MPC constraints. The envelope is not monotone with order: increasing the prediction order can temporarily increase the certified error before it decreases. This motivates explicit certification of each candidate order.

\begin{figure}[t]
    \vspace*{5pt}
    \centering
    \includegraphics[width=\linewidth]{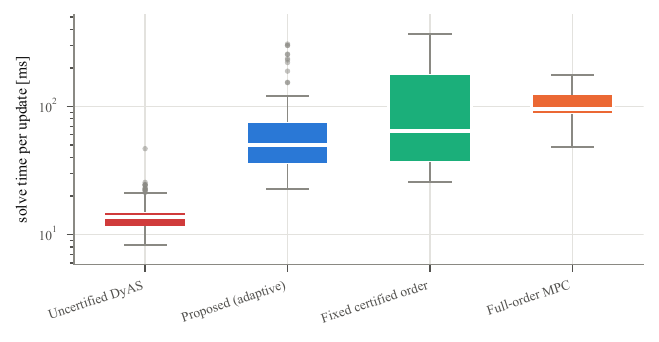}
    \caption{Computation time over $240$ updates for the four controllers.}
    \label{fig:solve-time-comparison}
\end{figure}

Figure~\ref{fig:solve-time-comparison} summarizes computation time over all $240$ updates. Uncertified DyAS is the least expensive because it avoids safety certification. The proposed controller incurs additional order-search and certification cost, but its typical computation time remains below the fixed certified-order controller and full-order MPC. The reported online cost includes the DyAS update, reduced-model and decoder evaluations, rejected solves, numerical verification, the exact first-step evaluation, and the accepted MPC solve. Overall, performance-only reduction is inexpensive but can be insufficient for safety, while high-order or full-order prediction is more costly. The proposed method adapts the prediction order to the level required for certified future safety.

\section{Conclusion}

This paper presented a safety-aware dimension-selection framework for reduced-order MPC. Dynamic active subspaces provide a performance-oriented ordering of reduced coordinates, while finite-horizon prediction-error bounds determine whether a given order is sufficient to maintain the future DCBF and terminal constraints. The resulting controller selects the smallest certified prediction order and increases model dimension only when required by the safety-critical operating condition. The proposed formulation guarantees full-order horizon safety, recursive feasibility, and closed-loop invariance of the safe set. Numerical results show that the performance-oriented order alone can be insufficient near safety-critical regimes, while the adaptive certified order increases when needed and reduces computation compared with consistently using a high-order or full-order model. The results also confirm that prediction error need not vary monotonically with order, motivating explicit certification of each candidate dimension.

\appendix

Proof of Theorem~\ref{thm:horizon}.

\begin{proof}
Let a feasible point of $\mathcal F_q(x_k)$ be fixed and let $x_{j|k}$ be the conceptual full-order trajectory under the same input sequence. From \eqref{eq:mpcfirst},
\begin{equation}
    x_{1|k}=F(x_k,u_{0|k})\in\D,
    \qquad
    h(x_{1|k})\ge\beta h(x_k).
    \label{eq:apph1}
\end{equation}
For every $j=1,\ldots,N$, Lemma~\ref{lem:error} gives
\begin{equation}
    \|x_{j|k}-R_q(z_{j|k})\|\le\bar r_{j,q}.
    \label{eq:apph2}
\end{equation}
Together with \eqref{eq:mpcdomain}, this implies
\begin{equation}
    x_{j|k}\in\D,
    \qquad j=1,\ldots,N.
    \label{eq:apph3}
\end{equation}
Lipschitz continuity of $h$ yields
\begin{align}
    h(x_{j|k})
    &\ge h(R_q(z_{j|k}))-L_h\bar r_{j,q},
    \label{eq:apph4}\\
    h(x_{j|k})
    &\le h(R_q(z_{j|k}))+L_h\bar r_{j,q}.
    \label{eq:apph5}
\end{align}
For $j=1,\ldots,N-1$, feasibility of \eqref{eq:mpcfuture} gives
\begin{align}
&h(R_q(z_{j+1|k}))-L_h\bar r_{j+1,q}\nonumber\\
&\qquad\ge
\beta\big(h(R_q(z_{j|k}))+L_h\bar r_{j,q}\big).
    \label{eq:apph6}
\end{align}
Hence
\begin{align}
    h(x_{j+1|k})
    &\ge h(R_q(z_{j+1|k}))-L_h\bar r_{j+1,q}\nonumber\\
    &\ge\beta\big(h(R_q(z_{j|k}))+L_h\bar r_{j,q}\big)\nonumber\\
    &\ge\beta h(x_{j|k}).
    \label{eq:apph7}
\end{align}
Combining \eqref{eq:apph1} and \eqref{eq:apph7},
\begin{equation}
    h(x_{j+1|k})\ge\beta h(x_{j|k}),
    \qquad j=0,\ldots,N-1.
    \label{eq:apph8}
\end{equation}
If $x_k\in\Ssafe$, then
\begin{equation}
    h(x_{j|k})\ge\beta^j h(x_k)\ge0,
    \qquad j=0,\ldots,N.
    \label{eq:apph9}
\end{equation}
Finally, if $\mathcal F_n(x_k)\neq\emptyset$, then the finite candidate set in \eqref{eq:qstar} is nonempty; therefore its minimum exists. This proves all claims.
\end{proof}

Proof of Theorem~\ref{thm:recursive}.

\begin{proof}
Fix a feasible order-$q$ point at time $k$, with input sequence
$\{u_{j|k}^\star\}_{j=0}^{N-1}$ and corresponding conceptual full-order trajectory
$\{x_{j|k}^\star\}_{j=0}^{N}$. Let $\ell\in\{1,\ldots,N\}$ of those inputs be applied in order. Determinism gives
\begin{equation}
    x_{k+s}=x_{s|k}^\star,
    \qquad s=0,\ldots,\ell.
    \label{eq:appr1}
\end{equation}
By Theorem~\ref{thm:horizon}, the stored full-order trajectory satisfies
\begin{align}
    x_{j|k}^\star&\in\D,
    \qquad j=0,\ldots,N,
    \label{eq:appr2}\\
    h(x_{j+1|k}^\star)&\ge\beta h(x_{j|k}^\star),
    \qquad j=0,\ldots,N-1.
    \label{eq:appr3}
\end{align}
The terminal containment \eqref{eq:mpcterm} and Lemma~\ref{lem:error} imply $x_{N|k}^\star\in\Xf$.
Starting from $y_0:=x_{N|k}^\star$, define the terminal continuation
\begin{equation}
    v_s:=\kappa_f(y_s),
    \qquad
    y_{s+1}:=F(y_s,v_s),
    \quad s=0,\ldots,\ell-1.
    \label{eq:appr5}
\end{equation}
Assumption~\ref{ass:terminal} yields, for every such $s$,
\begin{align}
    y_s&\in\Xf\subset\D,
    \label{eq:appr6}\\
    v_s&\in\Ksafe(y_s)\subseteq\U,
    \label{eq:appr7}\\
    y_{s+1}&\in\Xf,
    \qquad
    h(y_{s+1})\ge\beta h(y_s).
    \label{eq:appr8}
\end{align}
Now construct an $N$-step full-order candidate at time $k+\ell$ by taking the unused controls from the stored plan followed by the terminal continuation:
\begin{equation}
\tilde u_{j|k+\ell}:=
\begin{cases}
    u_{j+\ell|k}^\star,
    &0\le j\le N-\ell-1,\\
    v_{j-(N-\ell)},
    &N-\ell\le j\le N-1.
\end{cases}
\label{eq:appr9}
\end{equation}
The associated exact states satisfy
\begin{equation}
\tilde x_{j|k+\ell}=
\begin{cases}
    x_{j+\ell|k}^\star,
    &0\le j\le N-\ell,\\
    y_{j-(N-\ell)},
    &N-\ell\le j\le N.
\end{cases}
\label{eq:appr10}
\end{equation}
Equations \eqref{eq:appr2}--\eqref{eq:appr3} verify the shifted stored segment, while \eqref{eq:appr6}--\eqref{eq:appr8} verify the appended terminal segment. Since $\tilde x_{N|k+\ell}=y_\ell\in\Xf$, the terminal constraint holds. At $q=n$, the representation is exact and the error envelope vanishes by Assumption~\ref{ass:rom} and \eqref{eq:r1}--\eqref{eq:envelope}. Hence \eqref{eq:appr9} is feasible for the full-order problem, so $\mathcal F_n(x_{k+\ell})\neq\emptyset$. The same construction holds for $\ell=N$, when all $N$ inputs are generated by $\kappa_f$ from $x_{N|k}^\star\in\Xf$.
\end{proof}

\end{document}